\documentclass[journal,onecolumn]{IEEEtran}
\ifCLASSINFOpdf
   \usepackage[pdftex]{graphicx}
\else
\fi
\usepackage[cmex10]{amsmath}
\usepackage{amsthm}
\usepackage{amssymb}

\usepackage{algorithmic}

\usepackage{enumerate}

\newenvironment{Smallmatrix}[1]
  {\arraycolsep=2pt\footnotesize
   \array{#1}}
  {\endarray}
\newcommand{\ssmat}[1]{\left[\begin{Smallmatrix}{#1}}
\newcommand{\tamss}{\end{Smallmatrix}\right]}
\newcommand{\spmat}[1]{\left(\begin{Smallmatrix}{#1}}
\newcommand{\tamps}{\end{Smallmatrix}\right)}

\newcommand{\smat}[1]{\left[\!\begin{array}{#1}}
\newcommand{\tams}{\end{array}\!\right]}
\newcommand{\pmat}[1]{\left(\!\!\!\begin{array}{#1}}
\newcommand{\tamp}{\end{array}\!\!\! \right)}
\newcommand{\semat}[1]{\begin{Smallmatrix}{#1}}
\newcommand{\tames}{\end{Smallmatrix}}
\newcommand{\emat}[1]{\!\begin{array}{#1}}
\newcommand{\tame}{\end{array}\!}
\newcommand{\mcal}[1]{\mathcal{#1}}
\newcommand{\mbb}[1]{\mathbb{#1}}

\newcommand{\mbf}[1]{\mathbf{#1}}

\newtheorem{theorem}{Theorem}
\newtheorem{lemma}{Lemma}
\newtheorem{remark}{Remark}
\newtheorem{algorithm}{Algorithm}

\begin{document}
%
\title{Computation of lower bounds \\ for the induced $\mcal{L}_2$ norm of LPV systems}
%
%
%

\author{Tamas~Peni,~\IEEEmembership{Member,~IEEE,}
        and~Peter~J.~Seiler,~\IEEEmembership{Member,~IEEE}
\thanks{T. Peni is with the Systems and Control Laboratory of Institute for Computer Science and Control,  1111 Budapest, Kende u. 13-17., Hungary.
E-mail: peni.tamas@sztaki.mta.hu, pt@scl.sztaki.hu.}
\thanks{P. J. Seiler is with Aerospace Engineering and Mechanics, University of Minnesota, 107 Akerman Hall, 
110 Union St SE,  Minneapolis, MN 55455-0153. E-mail: eile017@umn.edu.}
\thanks{Manuscript, \today}}

%
%

\markboth{}%
{for arXiv}
%



\maketitle

\begin{abstract}
Determining the induced $\mcal{L}_2$ norm of a linear, parameter-varying (LPV)
system is an integral part of many analysis and robust control design
procedures.
 Most prior work
has focused on efficiently computing upper bounds for the induced $\mcal{L}_2$
norm.
The conditions for upper bounds are typically based on scaled
small-gain theorems with dynamic multipliers or dissipation
inequalities with parameter
dependent Lyapunov functions. This paper presents a complementary
algorithm to compute lower bounds for the induced $\mcal{L}_2$ norm. The
proposed approach
computes a lower bound on the gain by restricting the parameter
trajectory to be a periodic signal.
This restriction enables the use of recent results
for exact calculation of the $\mcal{L}_2$ norm for a periodic linear time varying 
system.  The
proposed lower bound algorithm has two benefits. First, the lower
bound complements standard upper bound techniques. Specifically, a
small gap
between the bounds indicates that further computation, e.g. upper
bounds with more complex Lyapunov functions, is unnecessary.
Second, the lower bound algorithm returns a "bad" parameter trajectory
for the LPV system that can be further analyzed
to provide insight into the system performance. 

\end{abstract}


%
\IEEEpeerreviewmaketitle

\section{Introduction}
%
%
%
%
\IEEEPARstart{D}{etermining} the induced $\mcal{L}_2$ norm of an LPV system is an integral part of many analysis and robust control design procedures. In general, this norm cannot be determined explicitly. Most prior work focuses on computing an upper bound on the induced $\mcal{L}_2$ norm. 
The method used to compute the induced gain upper bound depends primarily on the structure of the LPV system.  One class of LPV systems involves state matrices that are rational functions of the parameter.  In this case the LPV system can be expressed as a feedback interconnection (linear fractional transformation, LFT) of a known linear time invariant (LTI) system and a parameter dependent block.  The upper bound for these LFT-type LPV systems is typically computed using scaled mall gain theorems with multipliers and the full block S-procedure  \cite{scherer99_lminotes,apkarian95_ieeetac,scherer01_automatica}.  Another class of LPV systems involves state matrices with an arbitrary dependence on the parameter.  These systems, called "Gridded" LPV systems, can arise via linearization of nonlinear models on a grid of equilibrium operating points \cite{fenwu96_ijrnc,fenwu01_ijc,fenwu95_phd,rugh00_automatica,shamma96_handbook}.  The upper bound for Gridded LPV systems is computed using the dissipation inequality evaluated over a finite set of parameter grid points.  In both cases the upper bound calculation reduces to a convex optimization with linear matrix inequality (LMI) constraints.

This paper addresses the complementary problem, i.e. the approximation for the lower bound of the induced $\mcal{L}_2$ norm. Frozen point analysis is one simple way to compute a lower bound. Specifically, the induced $\mcal{L}_2$ norm of the LPV system over all possible parameter trajectories is clearly lower bounded by the restriction to constant (frozen) parameter trajectories.  A lower bound is thus obtained by computing the maximum gain of the LTI systems obtained by evaluating the LPV system on a grid of frozen parameter values.  To our knowledge this is currently the only method available to compute a lower bound on the induced $\mcal{L}_2$ norm of an LPV system. Unfortunately, this approach produces very conservative results in many cases as it neglects the variation of the scheduling parameter. To compute a better estimate, the $\mcal{L}_2$ norm has to be evaluated over time varying parameter trajectories. This concept introduces a complex optimization problem involving a maximization over both the allowable parameter trajectories and the $\mcal{L}_2$ inputs to the LPV system.  

The proposed lower bound algorithm restricts the scheduling trajectories to periodic signals. This restriction is useful because the LPV system evaluated on a fixed periodic trajectory is simply a periodic linear time-varying (PLTV) system. Moreover, the induced $\mcal{L}_2$ norm of a PLTV system can be exactly determined by using recently developed numerical methods, see e.g. \cite{cantoni09_automatica}.  Thus the lower bound algorithm only requires the maximization over periodic parameter trajectories and the maximization over $\mcal{L}_2$ inputs is not required.  Specifically, the lower bound algorithm maximizes a cost function related to the induced $\mcal{L}_2$ norm (described in detail in Sec. \ref{sec:lpvL2norm}) over a finite dimensional subspace of periodic scheduling trajectories. The output of the algorithm is the lower bound estimate and the worst-case (bad) parameter trajectory. In addition, we construct a worst-case input signal that approximately achieves the computed induced gain lower bound for the LPV system. The construction of the worst-case input is based on the PLTV results collected in \cite{colaneri00_kybernetica} and \cite{bittanticolaneri09_psbook}. In \cite{colaneri00_kybernetica} the worst case input is derived by using a special, frequency domain representation for the PLTV system. In this paper the worst-case input construction and proofs are provided using only time-domain operators.  This approach, detailed in section \ref{sec:wcinput}, streamlines the construction of the worst-case input.

It is important to note that our algorithm does not assume any specific structure on the LPV system. Thus it is developed for the "Gridded" case. It trivially applies for "LFT" models as well although the additional rational structure in LFT models may lead to faster lower bound algorithms than those developed here. It will also be shown, by numerical examples, that the proposed lower bound algorithm and the known upper bound methods complement each other. Specifically,  these methods, applied together, can yield a tight bound for the induced $\mcal{L}_2$ norm of an LPV system. 

The paper is organised as follows: in the next section the properties of periodic systems are summarized and the recent results related to the induced $\mcal{L}_2$ norm computation are reviewed. In Sec. \ref{sec:wcinput} we provide an algorithm to construct the worst-case input for PLTV systems. Our approach for computing lower bounds on $\mcal{L}_2$ norm of LPV systems is presented in Sec. \ref{sec:lpvL2norm} and \ref{sec:trajparam}. Sec. \ref{sec:examples} is devoted to the numerical simulations and analysis. The conclusions are drawn in section \ref{sec:conclusion}, where the future works required to improve our results are also summarized.

The notations used in the paper are fairly standard. $\mbb{R}$ and $\mbb{C}$ denote the fields of real and complex numbers, respectively. $\mbb{D}$ and $\partial\mbb{D}$ are the unit disc and unit circle in  the complex plane $\mbb{C}$. If $\mbf{T}:E \rightarrow F$ is a bounded, linear operator between Hilbert spaces $E$ and $F$, then the adjoint of $\mbf{T}$ is denoted by $\mbf{T}^*$ and is defined as follows: $\langle \mbf{T}u,y \rangle_F=\langle u,\mbf{T}^*y\rangle_E$, $\forall u\in E$ and $\forall y \in F$. The space of square-integrable signals $f:[0,\infty)\rightarrow \mbb{R}^n$ is denoted by $\mcal{L}_2(\mbb{R}^{n})$. The inner product and norm in $\mcal{L}_2(\mbb{R}^{n})$  are defined as $\langle f,g \rangle_{\mcal{L}_2(\mbb{R}^n)}=\int_0^\infty f(t)^Tg(t)dt$ and $\|f\|=\sqrt{\langle f,f \rangle_{\mcal{L}_2(\mbb{R}^{n})}}$, respectively. If $\mbf{G}$ is a bounded linear operator, such that $\mbf{G}:w\in\mcal{L}_2(\mbb{R}^{p})\mapsto z\in\mcal{L}_2(\mbb{R}^{q})$ then its induced $\mcal{L}_2$ norm is defined as 
\begin{align*}
\|\mbf{G}\|:=\sup_{0\ne w\in\mcal{L}_2(\mbb{R}^{p})} \frac{\|z\|}{\|w\|}
\end{align*}
$ \mcal{L}_{2,[0,h)}(\mbb{R}^{n})$ denotes the space of square-integrable functions on the interval $[0,h)$ with inner product  $\langle f,g \rangle_{\mcal{L}_{2,[0,h)}(\mbb{R}^n)}=\int_0^h f(t)^Tg(t)dt$. $\ell_2(E)$ denotes the square-summable sequences $w=\{w_k\}_{k=0}^\infty$ in the Hilbert space $E$ with inner product $\langle w,v \rangle_{\ell_2(E)}:=\sum_{k=0}^\infty \langle w_k,v_k\rangle_E$.


 

\section{Periodic systems}
\subsection{Background}
This section reviews known results on PLTV systems.  Most results can be found in \cite{bittanticolaneri09_psbook}.  In particular, we consider the linear time-varying system $G$ of the form:
\begin{align}
G: ~~ \emat{rcl} 
\dot{x}(t)&=&A(t)x(t)+B(t)w(t) \\
z(t)&=&C(t)x(t)+D(t)w(t)
\tame
\label{eq:pltv}
\end{align}
that has the following properties:
\begin{enumerate}[(i)]
\item the system matrices $A(t):\mbb{R}\rightarrow \mbb{R}^{n\times n}$, $B(t):\mbb{R}\rightarrow \mbb{R}^{n\times p}$, $C(t):\mbb{R}\rightarrow \mbb{R}^{q\times n}$, $D(t):\mbb{R}\rightarrow \mbb{R}^{q\times p}$ are bounded continuous functions of time, 
\item are $h$-periodic, i.e. $A(t+h)=A(t)$, $B(t+h)=B(t)$, $C(t+h)=C(t)$, $D(t+h)=D(t)$, $\forall t$ and 
\item the dynamics \eqref{eq:pltv} are internally stable.
\end{enumerate}
The state transition matrix $\Phi(t,\tau)$ associated with the autonomous dynamics $\dot{x}(t)=A(t)x(t)$ is defined for all $(t,\tau)$ as the linear mapping from $x(\tau)$ to $x(t)$, i.e. $x(t)=\Phi(t,\tau)x(\tau)$. It is easy to check that $\Phi(t,\tau)$ is periodic with $h$, i.e. $\Phi(t+h,\tau+h)=\Phi(t,\tau)$, and satisfies the matrix differential equation $ \frac{d}{dt}{\Phi}(t,\tau)=A(t)\Phi(t,\tau)$, $\Phi(\tau,\tau)=I$. The state transition matrix over one period is called the monodromy matrix: $\Psi(t):=\Phi(t+h,t)$.  Clearly,  the monodromy matrix is periodic with $h$, i.e. $\Psi(t+h)=\Psi(t)$. 

The system defined in \eqref{eq:pltv} is stable and hence it defines a bounded operator $\mbf{G}:=w\in\mcal{L}_2(\mbb{R}^{p})\mapsto z\in\mcal{L}_2(\mbb{R}^{q})$, such that
\begin{align}
z(t)=\int_0^tC(t)\Phi(t,\tau)B(\tau)w(\tau)~d\tau+D(t)w(t)
\label{eq:Gopdef}
\end{align}
i.e. $z(t)$ is the response of \eqref{eq:pltv} to the input signal $w(t)$ if $x(0)=0$. We are interested in  $\mbf{G}$. It is shown in \cite{bamieh92_ieeetac} that $\|\mbf{G}\|$ is equal to the norm of the "lifted" operator $\hat{\mbf{G}}:\ell_2(\mcal{L}_{2,[0,h)}(\mbb{R}^{p})) \rightarrow \ell_2(\mcal{L}_{2,[0,h)}(\mbb{R}^{q}))$, which, has a finite dimensional state-space realization in the following form \cite{bamieh91_scl}:
\begin{align}
\xi_{k+1}&=\hat A\xi_k+\hat B \hat w_k \nonumber \\
\hat z_k&=\hat C\xi_k +\hat D \hat w_k 
\label{eq:lifted-pltv}
\end{align}
where $\xi_k\in\mbb{R}^n$   and $\hat A: \mbb{R}^n \rightarrow \mbb{R}^n $, $\hat B:\mcal{L}_{2,[0,h)}(\mbb{R}^p)\rightarrow \mbb{R}^n$, $\hat C:\mbb{R}^n \rightarrow \mbb{R}^q $, $\hat D:\mcal{L}_{2,[0,h)}(\mbb{R}^p)\rightarrow \mcal{L}_{2,[0,h)}(\mbb{R}^q)$ are linear operators defined as follows:
\begin{align}
\hat A\xi_k&:=\Psi(0)\xi_k,&&
\hat B\hat w_k:=\int_0^h\Phi(h,\tau)B(\tau)\hat{w}_k(\tau)d\tau\nonumber\\
\hat C\xi_k&:=C(t)\Phi(t,0)\xi_k, &&
\hat D\hat w_k:=\int_0^t C(t)\Phi(t,\tau)\hat{w}_k(\tau)d\tau+D(t)\hat w_k(t)
\end{align}
Although \eqref{eq:lifted-pltv} is a finite dimensional system, it is still not suitable for numerical computations, because  its system matrices are operators.  Therefore, based on  the same idea applied  for LTI systems (see e.g. Chapter 21 in \cite{zhou96_roc}),  a discrete-time, linear, time-invariant system $\underline{G}_\gamma$ is introduced in \cite{dullerud99_scl},\cite{cantoni09_automatica}. The state-space matrices of $\underline{G}_\gamma$ are computed as follows:
\begin{align}
\underline A_\gamma &:= \hat{A}+\hat{B}(\gamma^2I-\hat D^*\hat D)^{-1}\hat D^*\hat C \nonumber\\
\underline B_\gamma\underline B_\gamma^*&:=\gamma \hat B(\gamma^2I-\hat D^*\hat D)^{-1}\hat B^*  \label{eq:AgBgCg} \\
\underline C_\gamma^*\underline C_\gamma&:=\gamma \hat C^*(\gamma^2I-\hat D\hat D^*)^{-1}\hat C \nonumber
\end{align} 
More specifically, the matrices $\underline{B}_\gamma$ and $\underline{C}_\gamma$ 
are defined to be full rank matrices that satisfy the equalities that appear in \eqref{eq:AgBgCg}. Note that $\underline A_\gamma,\underline B_\gamma,\underline C_\gamma$  are now real matrices $\underline A_\gamma\in\mbb{R}^{n\times n},\underline B_\gamma\in\mbb{R}^{n\times p},\underline C_\gamma\in\mbb{R}^{q\times n}$. The following theorem, taken from \cite{cantoni09_automatica}, proves that there is a strong relation between the induced $\ell_2$-norm of  $\underline G_\gamma$ and the norm of $\hat{\mbf{G}}$. 

\begin{theorem} \label{thm:cantoni1}
The following statements are equivalent:
\begin{enumerate}[a)]
\item $\text{eig}(\hat{A})\in\mbb{D} \backslash \mbb{\partial D}$  and $\|\hat{\mbf{G}}\|<\gamma$
\item $\text{eig}(\underline A_\gamma)\in\mbb{D} \backslash \mbb{\partial D}$ and  $\|\underline G_\gamma\|_\infty<1$
\end{enumerate}
where $\|\underline G_\gamma\|_\infty$ is the standard $\mcal{H}_\infty$ norm of the LTI system $\underline G_\gamma$.
\end{theorem}
\subsection{Induced norm computation}
Theorem \ref{thm:cantoni1}  gives the base of the bisection algorithm, proposed in \cite{cantoni09_automatica} to compute the $\mcal{L}_2$-norm of the periodic system \eqref{eq:pltv}. 
Therefore, an efficient method is needed to determine the system matrices $\underline A_\gamma,\underline B_\gamma,\underline C_\gamma$. For this, consider the following  differential equation associated with \eqref{eq:pltv}:
\begin{align}
\dot{e}(t)=H(t)e(t) 
\label{eq:HCantoni}
\end{align}
where (suppressing the notation of time dependence)
\begin{align}
H_{11}&= -A^T-C^TD(\gamma^2I-D^TD)^{-1}B^T \nonumber\\
H_{12}&= -\gamma C^T(\gamma^2I-DD^T)^{-1}C \nonumber \\
H_{21}&= \gamma B(\gamma^2I-D^TD)B^T \\
H_{22}&=A+B(\gamma^2I-D^TD)^{-1}D^TC \nonumber 
\end{align}
\noindent Note that \eqref{eq:HCantoni} is a Hamiltonian system ($JH(t)$ is symmetric with $J=\left[\begin{smallmatrix}0&I\\-I&0\end{smallmatrix}\right]$). Let the monodromy matrix of \eqref{eq:HCantoni} be denoted by $\Psi_H(t)$ and let $Q=\Psi_H(0)$.   It was shown in \cite{cantoni09_automatica}, that
\begin{align}
\underline B_\gamma\underline B_\gamma^T&=Q_{21}Q_{11}^{-1} \nonumber\\
\underline C_\gamma^T\underline C_\gamma&=-Q_{11}^{-1}Q_{12}\nonumber\\
\underline A_\gamma&=Q_{22}-Q_{21}Q_{11}^{-1}Q_{12} \label{eq:AgBgCg&Q}
\end{align}
where $Q=\left[\begin{smallmatrix}Q_{11}&Q_{12}\\Q_{21}&Q_{22}\end{smallmatrix}\right]$. Since $Q$ is symplectic \cite{neishtadt07_pres}, i.e. $Q^TJQ=J$, thus $Q$ is uniquely determined by \eqref{eq:AgBgCg&Q} in the following form:
\begin{align}
Q=\smat{cc}
\underline A_\gamma ^{-T} & -\underline A_\gamma ^{-T}\underline C_\gamma^T\underline C_\gamma \\ 
\underline B_\gamma\underline B_\gamma^T\underline A_\gamma ^{-T} & \underline A_\gamma-\underline B_\gamma\underline B_\gamma^T \underline A_\gamma ^{-T} \underline C_\gamma^T\underline C_\gamma
\tams 
\label{eq:QfromAgBgCg}
\end{align}
One possible method to determine the $\underline A_\gamma,\underline B_\gamma,\underline C_\gamma$ matrices can be given as follows. Integrate first the Hamiltonian system \eqref{eq:HCantoni} on $[0,h)$ starting from the matrix initial condition $e(0)=I$. Then $e(h)=Q$. Determine the system matrices from $Q$  by using \eqref{eq:AgBgCg&Q}. Since $H(t)$ is not stable this approach is numerically not reliable. The method proposed in \cite{cantoni09_automatica} is based on the following relations:
\begin{align*}
\underline A_\gamma=X(h),~~~\underline C_\gamma^T \underline C_\gamma=Z(0),~~~\underline B_\gamma\underline B_\gamma^T=Y(h)
\end{align*}
where $X(h),Z(0),Y(h)$ are point solutions of the differential Riccati equations
\begin{align}
\dot{Z}&=-H_{22}^TZ-ZH_{22}-\gamma ZB(\gamma^2I-D^TD)^{-1}B^TZ+\gamma C^T(\gamma^2I-DD^T)^{-1}C \label{eq:RicZ}\\
\dot{X}&=(H_{22}+\gamma B(\gamma^2I-D^TD)^{-1}B^TZ)X  \label{eq:RicX}\\
\dot{Y}&=H_{22}Y+YH_{22}^T+\gamma YC^T(\gamma^2I-DD^T)^{-1}CY+\gamma B(\gamma^2I-D^TD)^{-1}B^T  \label{eq:RicY}
\end{align}
\noindent with boundary conditions $Z(h)=0,X(0)=I$ and $Y(0)=0$. Integrating these Riccati equations is a numerically better conditioned problem than the direct integration of \eqref{eq:HCantoni} \cite{cantoni09_automatica}. Moreover, \eqref{eq:RicZ} has a very useful property \cite{cantoni09_automatica}:  
\begin{lemma}
Assume $\gamma^2I-D(t)^TD(t)>0$  for all $t \in [0,h)$. Then \eqref{eq:RicZ} has a bounded solution over the interval $[0,h]$ if and only if $\|\hat{D}\|<\gamma$. 
\end{lemma}
The bisection algorithm proposed in \cite{cantoni09_automatica} is based on iteratively solving \eqref{eq:RicZ}-\eqref{eq:RicY} and computing $\|\underline G_\gamma\|_\infty$ at different $\gamma$ values that are tuned in a bisection loop. The output is a lower- and upper bound pair $(\underline\gamma,\overline\gamma)$ satisfying $\underline\gamma\leq \|\mbf{G}\| \leq \overline\gamma$ such that $\overline\gamma-\underline\gamma\leq\varepsilon$, where $\varepsilon$ is a given tolerance. 

\begin{remark}
\noindent If we introduce $T=\left[\begin{smallmatrix}0 & I\\\gamma I & 0\end{smallmatrix}\right]$ and apply the state transformation $\tilde e(t)=T^{-1}e(t)$ in \eqref{eq:HCantoni}, we get the following transformed Hamiltonian system: 
\begin{align}
\dot{\tilde{e}}(t)=\tilde{H}\tilde e(t),~~\text{ with }~~\tilde{H}(t)=T^{-1}H(t) T.
\label{eq:HColaneri}
\end{align}
This system is the same as that is used in \cite{colaneri00_kybernetica}. 
Denote $\Psi_{\tilde H}(t)$ the  monodromy matrix associated with \eqref{eq:HColaneri} and let $\tilde Q:=\Psi_{\tilde H}(0)$. Then, it can be shown that  $\tilde Q=T^{-1}QT$. We use \eqref{eq:HColaneri} instead of \eqref{eq:HCantoni} in the next section, because  \eqref{eq:HColaneri}  is more convenient for the forthcoming derivations. 
\end{remark}
\subsection{Construction of the worst-case input}\label{sec:wcinput}
We are also interested in constructing the worst-case input $w^\circ \in \mathcal{L}_2$ that achieves the induced norm $\| \mathbf{G}\|$.  The signal achieving this gain is not in $\mathcal{L}_2$ and hence our construction approximately achieves this gain. Specifically, for any $\epsilon>0$ we construct an input $w^\circ \in \mathcal{L}_2$ such that $z^\circ := \mathbf{G}w^\circ$ satisfies $\|z^\circ\| \ge \| \mathbf{G}\| \|w^\circ \| - \epsilon$.  This is similar to the LTI case where the worst-case input is a sinusoid and a truncated sinusoid approximately achieves the system gain.    The method used to construct the worst-case input for the PLTV system is based on the proof of Lemma 2.6 in \cite{colaneri00_kybernetica}.  In \cite{colaneri00_kybernetica} the worst-case input is constructed using a special frequency-domain representation of $\mathbf{G}$ defined  over exponentially periodic signals.  In this section, an alternative construction and proof is provided using only time-domain formulations.  This alternative proof streamlines the numerical construction of the worst-case input.

\noindent To this end, let the linear operator $\mbf{T}_G:\mbb{R}^n\times\mcal{L}_{2,[0,h)}(\mbb{R}^p) \rightarrow \mbb{R}^n\times\mcal{L}_{2,[0,h)}(\mbb{R}^q)$ be defined as follows: 
\begin{align*}
\mbf{T}_G(x_0,w) \rightarrow (x_h, z):=(\hat Ax_0+\hat Bw,~\hat Cx_0+\hat D w). 
\end{align*}
$\mbf{T}_G$ is equivalent to the (state-input) $\rightarrow$ (state-output) map realized by \eqref{eq:pltv} over the period $[0,h)$:
\begin{align}
\dot{x}(t)&=A(t)x(t)+B(t)w(t),~~x(0)=x_0 \nonumber\\
z(t)&=C(t)x(t)+D(t)w(t),~~x_h=x(h)
\label{eq:LGss}
\end{align}
If we introduce an inner product in $\mbb{R}^n\times\mcal{L}_{2,[0,h)}(\mbb{R}^\cdot)$ as
\begin{align*}
\langle (x,w),(y,v) \rangle= x^*y+\int_0^h w^*(t)v(t)dt,
\end{align*}
then we can define the adjoint operator $\mbf{T}_G^*$ to satisfy  the equation
\begin{align*}
\langle (\hat x_h,\hat z), \mbf{T}_G(x_0,w) \rangle = \langle \mbf{T}_G^*(\hat x_h,\hat z),(x_0,w)\rangle.
\end{align*}
The next lemma shows how the state-space realization of the adjoint operator is related to the periodic system \eqref{eq:pltv}. 
\begin{lemma} \label{lem:innerprod}
If $(\hat x_0,\hat w)=\mbf{T}_G^*(\hat x_h,\hat z)$ then
\begin{align}
\dot{\hat x}(t)&=-A^*(t)\hat x(t)-C^*(t)\hat z(t),~~\hat x(h)=\hat x_h \nonumber \\
\hat{w}(t)&=B^*(t)\hat x(t)+D^*(t)\hat z(t),~~\hat x_0=\hat x(0) 
\label{eq:adjLGss}
\end{align}
\end{lemma}
\begin{proof} The proof can be found in the Appendix. \end{proof}
%

The following lemma provides a useful interpretation for the
Hamiltonian dynamics \eqref{eq:HColaneri}.

\begin{lemma}\label{lem:intercon} 
  Interconnect the dynamics of \eqref{eq:LGss} and \eqref{eq:adjLGss} with $\hat{z}(t):=z(t)$
  and $w(t):=\gamma^{-2} \hat{w}(t)$.  The resulting autonomous
  dynamics has state $\tilde{e}^T:=[ x^T \, \hat{x}^T ]$ with
  dynamics given by (suppressing dependence on $t$):
\begin{align}
  \label{eq:HColaneri2} 
  \dot{\tilde{e}} & = \tilde{H} \tilde{e}  \\
  \label{eq:w} 
  w & = (\gamma^{2} I - D^*D)^{-1} \smat{cc}  D^* C  & B^*\tams  \tilde{e} \\
  \label{eq:z} 
  z & = (\gamma^{2} I - DD^*)^{-1}  \smat{cc} \gamma^2 C & DB^*  \tams  \tilde{e} 
\end{align}
Moreover let $(\lambda,v)$ denote an eigenvalue/eigenvector of the
monodromy matrix $\tilde{Q}$ for (\ref{eq:HColaneri2}).  Partition
$v^*:=[v_1^* \, v_2^*]$ conformably with the state $\tilde{e}^T:=[ x^T
\, \hat{x}^T ]$. Then $\mathbf{T}_G(v_1,w) = (\lambda v_1,z)$ and
$\mathbf{T}_G^*(\lambda v_2,z) = (v_2, \gamma^{2} w)$.
\end{lemma}
\begin{proof}
  The expression for the Hamiltonian dynamics as the feedback
  connection of \eqref{eq:LGss} and \eqref{eq:adjLGss} is from  \cite{colaneri00_kybernetica}.  The expressions for
  $\mathbf{T}_G$ and $\mathbf{T}_G^*$ follow from the definitions of
  the forward and adjoint operators. 
\end{proof}

In what follows we show that worst-case input can be constructed from
(\ref{eq:w}) if the Hamiltonian system (\ref{eq:HColaneri2}) is initialized
as $\tilde{e}(0) = v$. For this we need the following theorem, which
is the direct application of Theorem 21.12 in \cite{zhou96_roc} to the
discrete-time system $\underline{G}_\gamma$: 

\begin{theorem} The following statements are equivalent:
\begin{enumerate}[(a)]
\item $\|\underline{G}_\gamma\|_\infty<1$.
\item $\tilde Q$ has no eigenvalues on the unit circle and $\|\underline{C}_\gamma(I-\underline A_\gamma)^{-1}\underline B_\gamma\|<1$
\end{enumerate}
\end{theorem}

It follows form Theorems 2.1 and 2.2 that if $\gamma \le
\|\mathbf{G}\|$ then $\tilde{Q}$ has a unit modulus eigenvalue
$e^{j\omega h}$ for some $\omega$.  Let $v$ denote the corresponding
eigenvector.  Let $\tilde{e}$ denote the solution of the
Hamiltonian dynamics (\ref{eq:HColaneri2}) and $(w,z)$ the corresponding
outputs with initial condition $e(0)=v$.  Then Lemma \ref{lem:intercon}
implies that $\mathbf{G}$ maps the input $w$ and initial
condition $v_1$ to the output $z$.  Moreover, it follows
from Lemma \ref{lem:intercon} that:
\begin{align}
  v_2^* v_1 + \int_0^h z^*(t) z(t) 
  & = \langle (e^{j\omega h} v_2,z), \mathbf{T}_G (v_1,w) \rangle \\
  & = \langle \mathbf{T}_G^*(e^{j\omega h} v_2,z), (v_1,w) \rangle \\
  & = v_2^* v_1 + \gamma^2 \int_0^h w^*(t) w(t) &
\end{align}
Thus $||z||_{[0,h)} = \gamma ||w||_{[0,h)}$, i.e. the $\mathcal{L}_2$
gain of $\mathbf{G}$ is equal to $\gamma$ on the interval $[0,h)$
provided the initial condition of the system $\mathbf{G}$ is given by
$x(0) = v_1$.

Note that the input/output pair $(w,z)$ for $t\in [0,h)$ is obtained
by integrating the Hamiltonian dynamics (\ref{eq:HColaneri2}) starting from
the initial condition $e(0)=v$. The state of the Hamiltonian system
after one period is given by $e(h) = \tilde{Q} e(0) = e^{j\omega h}
v$.  Thus integrating the periodic Hamiltonian dynamics forward over
the next interval yields $w(t) = e^{jwh}w(t-h)$ and $z(t) =
e^{jwh}z(t-h)$ for $t\in [h,2h)$.  Continuing to evolve the
Hamiltonian dynamics forward in time over subsequent periodic
intervals yields an input/output pair of $\mathbf{G}$ such that
for any interval $k\in \{0,1,2,\ldots \}$:
\begin{align}
\label{eq:wWC} 
  w(t) & = e^{j\omega kh} w(t-kh), \,\,\, t\in [kh,(k+1)h) \\
\label{eq:zWC} 
  z(t) & = e^{j\omega kh} z(t-kh) 
\end{align}
This input/output pair satisfies $\|z\|_{[kh,(k+1)h)} = \gamma
\|w\|_{[kh,(k+1)h)}$ over each interval $[kh,(k+1)h)$.  Thus $w$ is,
loosely speaking, an input that achieves the gain $\gamma$.  There are
three issues to be resolved to make this more precise.  First, the
input (\ref{eq:wWC}) is persistent and hence is not in
$\mathcal{L}_{2,[0,\infty)}$.  Second, the signals $(w,z)$ in
(\ref{eq:wWC}-\ref{eq:zWC}) are an input/output pair of $\mathbf{G}$
only if the periodic system starts with the initial condition
$x(0)=v_1$. However, the induced gain is defined assuming $x(0)=0$.
Third, the input can be complex if the eigenvector $v$ is complex.

The first two issues are resolved by noting that $(w,z)$ can be
expressed using the state transition matrix:
\begin{align}
\nonumber
z(t) = C(t) \Phi(t,0) v_1 
      & + \int_0^t C(t) \Phi(t,\tau)B(\tau)w(\tau) \, d\tau   + D(t) w(t)
\end{align}
Define $w^\bullet$ as the truncation of $w$ after $K$ periodic
intervals, i.e. $w^\bullet(t) := w(t)$ for $t\in[0,Kh)$ and
$w^\bullet(t):=0$ otherwise. Define $z^\bullet$ as the output of
$\mathbf{G}$ driven by input $w^\bullet$ but with initial condition
$x(0)=0$. By causality, $z^\bullet(t)=z(t)-s(t)$ for
$t<Kh$ where $s(t):=C(t) \Phi(t,0) v_1$ is the initial condition
response.  The norm of $z^\bullet$ can be bounded as:
\begin{align*}
  \| z^\bullet \| &  \ge \| z \|_{[0,Kh)} - \| s \|_{[0,Kh)} \\
                & = \gamma \| w^\bullet \|_{[0,Kh)} - \| s \|_{[0,Kh)} 
\end{align*}
The first line follows from the triangle inequality.  The second line
follows from two facts. First, $\|z\|_{[kh,(k+1)h)} = \gamma
\|w\|_{[kh,(k+1)h)}$ over each interval as noted above.  Second,
$w^\bullet=w$ by construction for $t\in [0,Kh)$.  
Next note that $ \| s \|_{[0,Kh)} \le \| s \| < \infty$ due
to the stability of $\mathbf{G}$.  In addition,
$\| w^\bullet \|_{[0,Kh)} = K \| w^\bullet \|_{[0,h)}$ and hence
$\| w^\bullet \|_{[0,Kh)} \rightarrow \infty$ as $K\rightarrow \infty$.
Thus it is clear that for all $\epsilon >0$ there exists an integer
$K$ such that the input-output pair $(w^\bullet,z^\bullet)$ of
$\mathbf{G}$ satisfies $\frac{\|z^\bullet\|}{\|w^\bullet\|} \ge \gamma - \epsilon$.
The input $w^\bullet$ is a valid ``worst-case'' signal because
it is in $\mathcal{L}_2$ and the output has been generated with
zero initial conditions.

The remaining issue is the fact that $(w^\bullet,z^\bullet)$ may be
complex. The periodic system $\mathbf{G}$ is linear and the system
matrices are real. Thus neglecting the imaginary part of the complex
valued input does not change the gain, i.e. the ratio of the norms of
the real input and the corresponding real output remains the
same. Therefore the real valued worst-case input can be obtained
by defining $w^\circ = Re(w^\bullet)$. Now we can summarize the complete
algorithm.


\begin{algorithm} [Worst-case input]\label{alg:wcinput} $ $\\
\vspace*{-0.3cm}
\begin{algorithmic}[1]
\STATE Let $\underline\gamma$ be given such that $\underline\gamma\leq\|\mbf{G}\|$.  Set $\gamma:=\underline\gamma$.
\STATE Compute $Q$ from $A_\gamma,B_\gamma,C_\gamma$ using the equation \eqref{eq:QfromAgBgCg}. Determine $\tilde Q$ by applying the similarity transformation $T$.

\STATE Compute the unit modulus eigenvalue $e^{j\omega h}$ of $\tilde Q$ and
  the corresponding eigenvector $v=\smat{cc} \overline x^* & \overline{\hat x}^* \tams^*$.

\STATE Integrate  the Hamiltonian dynamics \eqref{eq:HColaneri} from $t=0$ to $t=h$ starting from the initial condition $\tilde{e}(0)=v$.
\STATE Use \eqref{eq:w} to compute $w(t)$ from the state trajectories $x(t),\hat x(t)$ over $[0,h)$. 
\STATE Choose an integer $K\gg 1$ and define the (possibly complex valued)
  signal $w^\bullet(t)$ by:
  \begin{align*}
    w^\bullet(t)=
    \begin{cases}
     e^{j\omega kh}w(t-kh)   & \text{ if } t\in[kh,(k+1)h),\\
     &~~~k=0,1,\ldots,K-1 \\
          0 & \text{ if } t \ge Kh
    \end{cases}
  \end{align*} 
\STATE Let $w^\circ :=\text{Re}(w^\bullet)$ and $z^\circ := \mbf{G} w^\circ$.
\end{algorithmic}
\end{algorithm}

\section{Computation of lower bounds for the induced $\mcal{L}_2$ norm of LPV systems }\label{sec:lpvL2norm}
\subsection{Problem formulation}
Linear parameter varying (LPV) systems are a class of systems whose state space matrices depend on a time-varying parameter vector  $\rho$.  We assume that $\rho:\mathbb{R}^+ \rightarrow  \mathbb{R}^{m}$
is a  piecewise continuously differentiable function of time and satisfies the known bounds:
\begin{align}
\underline{\rho_i}\leq \rho_i(t) \leq \overline{\rho}_i,~~\underline{\mu_i}\leq \dot\rho_i(t) \leq \overline{\mu}_i,~1\leq i \leq m
\label{eq:rhobnds}
\end{align}
The set of allowable parameter vectors, denoted $\mathcal{P} \subseteq \mbb{R}^m$ consists of vectors $\rho$ that satisfy the range bounds given in \eqref{eq:rhobnds}. The set of admissible trajectories, denoted $\mathcal{A}$, consists of piecewise continuously differentiable trajectories that satisfy both the rate and range bounds given also in \eqref{eq:rhobnds}.

The state-space matrices of an LPV system are continuous functions of the parameter:  $A:
\mathcal{P} \rightarrow \mathbb{R}^{n \times n}$, $B: \mathcal{P} \rightarrow
\mathbb{R}^{n \times p}$, $C: \mathcal{P} \rightarrow \mathbb{R}^{q \times n}$ and
$D: \mathcal{P} \rightarrow \mathbb{R}^{q \times p}$.  An $n^{\textrm{th}}$ order LPV system, $G$, is defined by the following state-space form:
\begin{equation}
    \label{eq:lpv}
    \begin{split}
      \dot{x}(t) &= A(\rho(t)) x(t) + B(\rho(t))w(t) \\
      z(t) &= C(\rho(t)) x(t) + D(\rho(t))w(t)
    \end{split}
  \end{equation}
The performance of an LPV system $G$ can be specified in terms of its induced $\mathcal{L}_2$ gain from input $w$ to output $z$. The induced $\mathcal{L}_2$-norm is defined by
  \begin{equation}
    \|G\| = \sup_{0\neq w \in
      \mathcal{L}_2(\mbb{R}^p), \rho(\cdot) \in \mathcal{A}} \frac{ \|z\| }{ \|w\| },
      \label{eq:lpvL2gain}
  \end{equation}
The initial condition is assumed to be $x(0) = 0$. The notation $\rho(\cdot) \in \mathcal{A}$ refers to the entire (admissible) trajectory as a function of time. The class of LPV system given above has an arbitrary dependence on the parameter. For this class of systems there are known linear matrix inequality (LMI) conditions to efficiently compute an upper bound on the induced $\mcal{L}_2$ gain \cite{fenwu01_ijc}, \cite{fenwu95_phd}.

Now we address the complementary problem, i.e. our aim is to determine a lower bound for the induced  $\mcal{L}_2$ norm by using the results of the previous section.  

\subsection{Lower bound for the induced $\mcal{L}_2$ norm}
The computation of the lower bound is based on restricting the scheduling parameter trajectories to a finite-dimensional set of periodic signals. Let $\rho: \mbb{R}^+ \times \mbb{R}^N \rightarrow \mbb{R}^{m}$ denote a function that specifies a periodic scheduling trajectory for each value of $c\in \mbb{R}^N$.  $\rho(\cdot,c)$ denotes the entire trajectory (as a function of time) at the particular value $c$ and $\rho(t,c)$ denotes the $m$-dimensional scheduling vector obtained by evaluating $\rho(\cdot,c)$ at time instant $t$.  The trajectory is assumed to be periodic, i.e. for each $c$ there is a period $h(c)$ such that $\rho(t+h(c),c)=\rho(t,c)$ $\forall t$.  In addition, we must ensure the trajectory is admissible in the sense that it satisfies the range and rate bounds in \eqref{eq:rhobnds}. Let $\mathcal{C}_p \subseteq \mbb{R}^N$ denote the set of values that lead to such admissible trajectories, i.e. $\rho(\cdot,c) \in \mcal{A}$ for all $c\in \mcal{C}_p$.  The corresponding set of periodic trajectories is defined as
\begin{align}
  \mathcal{A}_p := \left\{ \rho(\cdot,c) \, | \, c \in \mathcal{C}_p \right\}
\end{align}
As a concrete example, $\rho$ can be specified as a linear combination of periodic bases functions
$\{\phi_k\}_{k=1}^N$, i.e.  $\rho(\cdot,c) := \sum_{k=1}^N c_k \phi_k(\cdot)$. Section \ref{sec:trajparam} provides alternative characterizations to specify periodic trajectories. 

Define the lower bound $\gamma_{lb}$ on the induced $\mathcal{L}_2$ gain \eqref{eq:lpvL2gain} as
\begin{align}
\label{eq:gammalb}
\gamma_{lb}:= \sup_{\rho(\cdot,c) \in \mathcal{A}_p} \| \mbf{G}_{\rho(\cdot,c)} \|=\sup_{c \in \mathcal{C}_p} \| \mbf{G}_{\rho(\cdot,c)} \|
\end{align}
where $\mathbf{G}_{\rho(\cdot,c)}$ denotes the periodic system (operator) obtained by evaluating the LPV system $G$ along the periodic trajectory specified by $\rho(\cdot,c)$.  The algorithm described in \cite{cantoni09_automatica} can thus be used to evaluate the gain $\| \mbf{G}_{\rho(\cdot,c)}\|$.  It follows immediately from $\mathcal{A}_p \subset \mathcal{A}$ that $\gamma_{lb} \le \|G\|$.  Equation \eqref{eq:gammalb} defines a finite dimensional optimization problem, which is non-convex in general. 
One further issue is that a single, accurate evaluation of $\| \mbf{G}_{\rho(\cdot,c)} \|$ requires many bisection steps and, as a consequence the matrix differential equations (Equations \eqref{eq:RicX}, \eqref{eq:RicY} and \eqref{eq:RicZ}) must be integrated many times for a single evaluation of the objective function.   Thus the evaluation of $\|\mbf{G}_{\rho(\cdot,c)}\|$ is computationally costly. A significant reduction in computation time can be achieved by using Algorithm \ref{alg:L2lowbnd-I} described below.  In this algorithm, $\underline{G}_{\gamma,\rho(\cdot,c)}$ denotes the discrete-time system \eqref{eq:AgBgCg} corresponding to the PLTV operator $\mathbf{G}_{\rho(\cdot,c)}$.  Moreover, define $\nu(c,\gamma):=\| \underline{G}_{\gamma,\rho(\cdot,c}\|$. With this notation, $\| G_{\rho(\cdot,c)}\| < \gamma$ if and only if $\nu(c,\gamma)<1$. The lower bound algorithm can now be stated.

\begin{algorithm}[Lower bound on $\mcal{L}_2$ norm-I] \label{alg:L2lowbnd-I}  $ $\\
\vspace*{-0.4cm}
\begin{algorithmic}[1]
\STATE Pick an initial parameter vector $c_0$. Let $c:=c_0$. 
\STATE Compute the norm of $\mbf{G}_{\rho(\cdot,c_0)}$ by using the algorithm in \cite{cantoni09_automatica}. Take the upper bound $\overline\gamma$ from the  bisection. It is clear that $\nu(c_0,\bar\gamma)<1$ because $\bar\gamma$ is an upper bound for $\|\mbf{G}_{\rho(\cdot,c_0)}\|$. 
\STATE  Starting from the initial value $c_0$, solve the  nonlinear optimization problem
\begin{align}
\sup_{c\in\mcal{C}_p} \nu(c,\bar\gamma)
\label{eq:maxnu1}
\end{align}
Three different outcomes are possible: \eqref{eq:maxnu1} terminates at a (local) optimum $c^*$ where $\nu(c^*,\gamma)<1$  {\it(case-a)}, $\nu(c^*,\gamma)\ge1$ {\it(case-b)},  or the optimization stops because at some $c=c^*$  the solution $Z(t)$ of \eqref{eq:RicZ} goes unbounded {\it(case-c)}.  Clearly, {\it case-b} and {\it case-c} mean that $\bar\gamma$ is smaller than $\|\mbf{G}_{\rho(\cdot,c^*)}\|$  so the optimization managed to improve the lower bound.  Therefore, set $c_0:=c^*$ and go to step 2. In {\it case-a} the optimization above was unsuccessful in the sense that the norm was not significantly improved. In this case, go to step 4.
\STATE  Compute the norm of $\mbf{G}_{\rho(\cdot,c^*)}$ and take the lower bound  $\underline\gamma$ from the  bisection. Let $\gamma_{lb}:=\underline\gamma$ and stop.
\end{algorithmic}
\end{algorithm}
The advantages of Algorithm \ref{alg:L2lowbnd-I} over the direct maximization of $\|\mbf{G}_{\rho(\cdot,c)}\|$ are the following: the bisection has to be performed less times (only once before each optimization step \eqref{eq:maxnu1} and then once at the end of the procedure) and the computation of $\nu$  requires only a single integration of the matrix differential equations, i.e.   it can be evaluated with significantly less computation than $\|\mbf{G}_{\rho(\cdot,c)}\|$. Testing the algorithm on numerical examples,  we found in most cases that only a few or even just one optimization step is enough to get a good approximation for the lower bound. The problem is often with the solution of the Riccati equation \eqref{eq:RicZ}, which goes unbounded when the norm of $\mbf{G}_{\rho(\cdot,c)}$  becomes significantly larger than $\bar\gamma$  at some $c=\hat c$. If this happens, a new $\bar\gamma$ has to be computed by performing the time-consuming bisection algorithm. This situation  can be avoided if \eqref{eq:maxnu1} is started with some $\bar\gamma\gg \|\mbf{G}_{\rho(\cdot,c_0)}\|$. Clearly if $\bar\gamma$ is chosen to be an upper bound of the $\mcal{L}_2$ norm of the LPV system, then $\|\mbf{G}_{\rho(\cdot,c)}\|\leq \|G\|\leq\bar\gamma$ for all $c$, which implies that $Z(t)$ never goes unbounded.
The modified, 1-step algorithm is summarized as follows:

\begin{algorithm}[Lower bound on $\mcal{L}_2$ norm-II] \label{alg:L2lowbnd-II}$ $\\
\vspace*{-0.4cm}
\begin{algorithmic}[1]
\STATE Compute an upper bound  $\gamma_{ub}$  on the gain of the LPV system, i.e. $\gamma_{ub} \ge \| G\|$.  Such an upper bound can be determined by using standard methods based on dissipativity relation (e.g. Bounded Real  type LMI conditions) \cite{fenwu01_ijc}, \cite{fenwu95_phd}. Let $\bar\gamma:=\gamma_{ub}$. Pick an initial parameter vector $c_0$.
\STATE Solve the optimization problem \eqref{eq:maxnu1}. 
\STATE Compute the norm of $\mbf{G}_{\rho(\cdot,c^*)}$ and take $\underline\gamma$ from the outputs of the bisection. Let $\gamma_{lb}:=\underline\gamma$ and stop.
\end{algorithmic}
\end{algorithm}
This algorithm only requires the $\mcal{L}_2$ nor of the periodic system to be evaluated at the last step.  This greatly reduces the number of required integrations for the matrix differential equations associated with X, Y, and Z. By performing either of the algorithms above we obtain a lower bound $\gamma_{lb}$ for  the induced $\mcal{L}_2$ gain and the worst-case scheduling trajectory $\rho(\cdot,c^*)$, where the associated PLTV system takes this norm. By using Algorithm \ref{alg:wcinput} we can also construct a worst-case input signal $w^o$ for $\mbf{G}_{\rho(\cdot,c^*)}$.  This, together with $\rho(\cdot,c^*)$ gives the (worst-case input, worst-case scheduling trajectory) pair, where the $\mcal{L}_2$-gain of the LPV system is $\gamma_{lb}$. 

\subsection{Scheduling trajectories}\label{sec:trajparam}

Algorithms  \ref{alg:L2lowbnd-I} and \ref{alg:L2lowbnd-II} optimize the $\mcal{L}_2$  bound over the elements of $\mcal{A}_p$. Therefore it is  important how this set is characterized, i.e. how the periodic signals depend on the parameter vector $c$. The structure of $\rho(\cdot,c)$  influences the convergence properties of \eqref{eq:maxnu1} and determines the final result. Therefore it has to be carefully chosen.  There are many ways to construct periodic signals. In this section we focus on piecewise linear scheduling trajectories, while another construction method, based on sinusoidal basis functions, is described in \cite{peni14_accinprep}.  Both methods result in a polytopic parameter set, i.e. $\mcal{C}_p$ is defined by linear inequality constraints. The linearity of the constraints simplifies the optimization \eqref{eq:maxnu1}.

Let the scheduling trajectories be specified by restricting the time derivative $\dot{\rho}$ to be piecewise constant.  Specifically, consider the one-dimensional case ($n_\rho=1$) where the time interval $[0,h]$ is subdivided into $R$ intervals $[T_{j-1}, T_j)$ for $j=1,\ldots,R$ with $T_0:=0$ and $T_R=h$.  Let $c_j := T_j - T_{j-1}$ denote the length of the $j^{th}$ sub-interval.  The parameter trajectory is then given as the periodic function of period $h$ defined for $t\in [0,h]$ by:
\begin{align}
\rho(0,c) & := \rho_0 \\
\label{eq:rhodot}
\dot{\rho}(t,c) & :=r_j  \mbox{ for } t \in [T_{j-1}, T_j)
\end{align}
where $r_j \in [\underline{\mu},\bar{\mu}]$ are user-selected values that specify the rate on each interval.  The vector $c:=[\rho_0,c_1, \ldots, c_R] \in \mbb{R}^{R+1}$ specifies the parameter trajectory offset $\rho_0$ and the subinterval lengths.  Equation~\eqref{eq:rhodot} can be integrated to give the explicit form of the trajectory as:
\begin{align}
  \rho(t,c) & = \rho_0 + \sum_{i=1}^{j-1} c_i r_i + (t-T_{j-1}) r_j ~~    \mbox{ for } t \in [T_{j-1}, T_j)
\end{align}
Note that the period $h=c_1+ \cdots + c_R$ is a free parameter as it depends on the interval lengths.  Moreover, $\rho(t,c)$ is a linear function of the offset $\rho_0$ and interval lengths $c_j$.  Thus the magnitude constraints in \eqref{eq:rhobnds} can be transformed into $R+1$ linear constraints on $c$:
\begin{align}
  \underline{\rho} \le \rho_0 + \sum_{i=1}^{j} c_i r_i  \le \bar{\rho}, ~~~~~  j \in \{ 0,\ldots,R\} 
  \label{eq:trajconstr}
\end{align}
To ensure $\rho(0,c)=\rho(h,c)$, \eqref{eq:trajconstr} has to be completed with an equality condition  $\sum_{i=1}^R c_ir_i=0$. Since all constraints are linear, thus $\mathcal{C}_p$ is a polytope.

\begin{remark}
The construction of the worst-case input requires additional considerations.  In particular, Algorithm  \ref{alg:wcinput}  for constructing the worst-case input requires integrating the unstable Hamiltonian system over $[0,h)$.  This causes numerical issues if the period length $h$ is "too large".  Thus it is reasonable to introduce an inequality constraint of the form $h\le \bar{h}$ to bound the period length when performing the optimization (24). This is simply another linear constraint on the vector $c$.
\end{remark}

\section{Numerical examples}\label{sec:examples}
In this section three numerical examples are presented to demonstrate the applicability of the proposed methods. To initialize our algorithms we need to determine an upper bound $\gamma_{ub}$ for the induced $\mcal{L}_2$ gain. In all examples $\gamma_{ub}$ is computed by solving the following optimization problem:
\begin{align}
\min_{V(x,\rho)} \gamma \nonumber \\
V(x,\rho)>0, \dot{V}(x,\rho,\dot\rho)\leq \gamma^2w^Tw-z^Tz
\label{eq:griddedL2gain}
\end{align}
where the Lyapunov (storage) function $V(x,\rho)$ was chosen to be quadratic: $V(x,\rho)=x^TP(\rho)x$, $P(\rho)=P(\rho,P_0,P_1,\ldots,P_M)$, where $P_i$-s denote the free (matrix) variables to be found. The infinite LMI constraints obtained were transformed to a finite set by choosing a suitable dense grid over the parameter domain $\mcal{P}$ and only the inequalities evaluated at the grid points are considered \cite{fenwu96_ijrnc}. 

Our strategy to compute the lower bound on the $\mcal{L}_2$ gain was the following: first we performed Algorithm \ref{alg:L2lowbnd-II} and then we used the parameter trajectory and lower bound we obtained as an initial value to run Algorithm \ref{alg:L2lowbnd-I}. Since in all numerical examples this second optimization  cannot improve the previous results, thus the numerical results we described in the forthcoming sections are generated by  Algorithm \ref{alg:L2lowbnd-II}. The nonlinear optimization was performed under MATLAB, by using the pattern search algorithm implemented in the Global Optimization Toolbox.

\subsection{LPV system with gain-scheduled PI controller}\label{sec:harald}
The first example, taken from \cite{pfifer14_ijrnc} is a feedback interconnection of a first-order LPV system with a gain-scheduled proportional-integral controller. The state-space matrices of the closed-loop system are as follows
\begin{align}
&A(\rho):=\smat{cc}-\frac{1}{\tau(\rho)}(1+K_p(\rho)K(\rho)) & \frac{1}{\tau(\rho)}\\ -K_i(\rho)K(\rho) & 0\tams,  
&&B(\rho):=\smat{c} \frac{1}{\tau(\rho)}K_p(\rho) \\ K_i(\rho) \tams, \nonumber\\
&C(\rho):=[-K(\rho)~0], &&D:=1
\label{eq:haraldscloop}
\end{align} 
where $\tau(\rho):=\sqrt{13.6-16.8\rho}$, $K(\rho):=\sqrt{4.8\rho-8.6}$ and 
\begin{align*}
K_p(\rho)=\frac{2\xi_{cl}\omega_{cl}\tau(\rho)-1}{K(\rho)},~~
K_i(\rho)=\frac{\omega_{cl}^2\tau(\rho)}{K(\rho)} ,~~\xi_{cl}=0.7,~~\omega_{cl}=0.25.
\end{align*}
The scheduling parameter is assumed to vary in the interval $[2,7]$ and $\dot{\rho}\in [-1,~1]$. 
\noindent By performing the optimization \eqref{eq:griddedL2gain} with 
\begin{align*}
V(x,\rho)=x^T\left(P_0+\sum_{k=1}^6 \rho^k P_k+\frac{1}{\rho}P_7+\frac{1}{\rho^2}P_8+\frac{1}{\rho^3}P_9\right)x\\
\Gamma:=\{\rho_1=2,\ldots,\rho_{100}=7\},~~\rho_{k+1}-\rho_k=5/99
\end{align*} 
we got $\gamma_{ub}=2.964$ for the upper bound. Using the parameter values in $\Gamma$ the frozen lower bound was also computed as $\gamma_{lb,fr}=\max_k \|G_{\rho_k}\|$, $\rho_k\in\Gamma$, where $G_{\rho_k}$ denotes the LTI system obtained by substituting $\rho(t)=\rho_k$ for all $t$.  The lower bound we obtained is $\gamma_{lb,fr}=1.1066$. To compute the lower bound by Algorithm \ref{alg:L2lowbnd-II}, we chose the following pattern for the rate variation of the scheduling trajectory  $r=[\overline\mu~0~\underline\mu~0~\overline\mu~0]=[1~0~-1~0~1~0]$ and let the algorithm tune the period in the interval $[h_{min}~h_{max}]=[12~20]$. The lower bound we obtained is $\gamma_{lb,pwl}=2.84$, at $h=15.877$. The scheduling trajectory can be seen in Fig. \ref{fig:haraldswc}. This lower bound is is significantly larger than $\gamma_{lb,\text{\cite{peni14_accinprep}}}$, where  $\gamma_{lb,\text{\cite{peni14_accinprep}}}$ is the lower bound we obtained in \cite{peni14_accinprep} by using sinusoid scheduling trajectories. Moreover, $\gamma_{lb,pwl}$ is very close to $\gamma_{ub}$: the difference  is only $\gamma_{ub}-\gamma_{lb,pwl}=0.13$. This means that we have very tight bounds on the norm of $G$: $2.84=\gamma_{lb,pwl}\leq\|G\|\leq \gamma_{ub}=2.964$.  

Finally, by using Algorithm  \ref{alg:wcinput} we computed also the worst case input for both scheduling trajectories. The input signals were constructed starting from the unit modulus eigenvalue $-1$. The parameter $K$ in Algorithm \ref{alg:wcinput} was  chosen to be $K=60$. The result can be seen in Fig. \ref{fig:haraldswc}.

\begin{figure}
\begin{center}
\includegraphics[width=12cm]{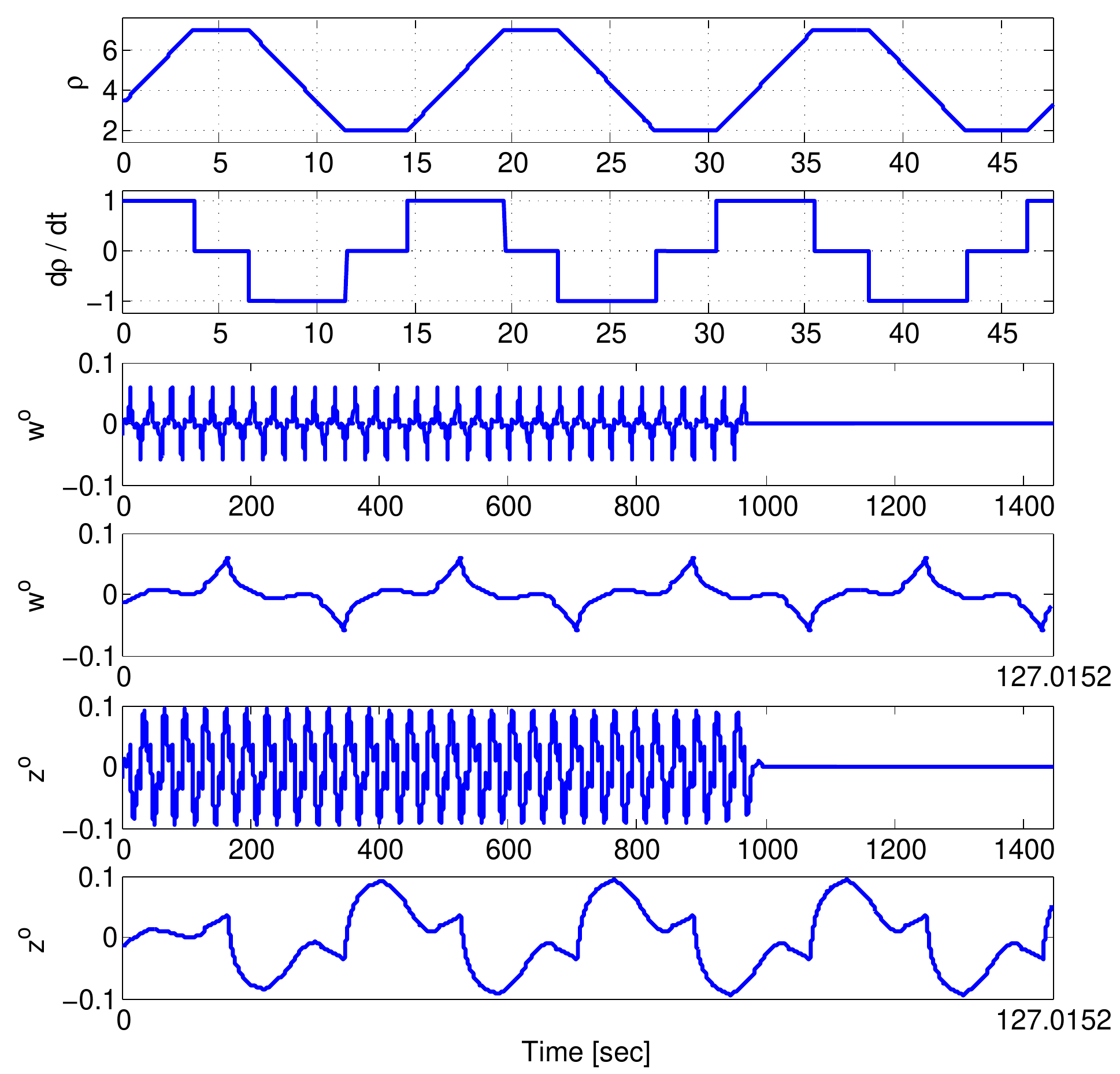}
\end{center}
\vspace*{-0.4cm}
\caption{Worst case scheduling trajectories and worst-case input/output for the LPV system in Section \ref{sec:harald}. The first two figures show the trajectory of the worst-case scheduling parameter $\rho(t)$ and its time derivative $\dot\rho(t)$.  In the 3rd and 4th figures $w^o$ is depicted over a longer and a shorter ($[0~8h]$) time intervals. In the 5th and 6th figures $z^o$ is displayed over the same time intervals. }
\label{fig:haraldswc}
\end{figure}

\subsection{Input and output scaled LTI system} \label{sec:catws}
The next example was constructed by taking two copies of the simple LTI system $1/(s+1)$, scaling the input of the first and he output of the second and computing the difference of the two outputs. The dynamics are given as follows:
\begin{align*}
\smat{c}\dot{x}_1\\ \dot{x}_2\tams &=\smat{cc}-1&0\\0&-1\tams \smat{c}x_1\\ x_2\tams+\smat{c}1\\\rho\tams w \\
z&=\smat{cc}\rho & -1\tams \smat{c}x_1\\ x_2\tams
\end{align*} 
We assume that $\rho\in [-1~~1]$ and $\dot\rho\in[-\overline\mu~~\overline\mu]$. We are going to compute the lower bound for different values of $\overline\mu$. It follows from the structure of the system that if $\rho$ is constant then the difference between the input and the output scaled systems is 0. This implies that $\gamma_{lb,fr}=0$. Next, the upper bound on the $\mcal{L}_2$ gain was computed for different rate bounds. In all cases the storage function $V(x,\rho)$ and the parameter grid $\Gamma$ was chosen as follows:
\begin{align*}
V(x,\rho)&=x^T\left(P_0+\sum_{k=1}^{10} \rho^k P_k\right)x\\
\Gamma&=\{\rho_1=-1,\ldots,\rho_{100}=1\},~~\rho_{k+1}-\rho_k=2/99
\end{align*}
The upper bounds obtained are collected in Table \ref{tab:catwsresults}.  To compute the lower bound we used the following rate pattern to characterize the piecewise linear scheduling trajectories: $r=[\overline\mu~0~-\overline\mu~0~\overline\mu~0]$. The lower bounds and the period of the worst-case scheduling trajectories obtained for the different $\overline\mu$ values are collected in Table \ref{tab:catwsresults}.  The range bounds $h_{\min}$ and $h_{\max}$ used to constrain the period length in \eqref{eq:maxnu1} are also given in Table \ref{tab:catwsresults}. It can be seen that the upper and lower bounds are very close to each other, which means that by using the upper and lower bound algorithms we could precisely compute the norm of this LPV system. In two particular cases, $\overline\mu=1$ and $\overline\mu=2$, we plotted also the worst case scheduling trajectory and the worst-case input in Fig. \ref{fig:catwswc}.  (The worst case inputs were computed starting from the unit modulus eigenvalues $[0.9723 + 0.2337i]$(in case $\overline\mu=1$) and $[0.9704 + 0.2415i]$ (in case $\overline\mu=2$). $K$ was 60 in both cases.) 


\begin{table}[h!]
\begin{align*}
\emat{c|c|c|c|c|c|c|c|c|c|c|c|c}
\overline\mu & \gamma_{lb} & \gamma_{ub} & \gamma_{ub}-\gamma_{lb}&[h_{min}~h_{max}] & h \\\hline
0.1 & 0.0979 & 0.1087 & 0.0108 & [0.5~50] & 49.9 \\
0.4 & 0.3309 & 0.3342 & 0.0033 & [0.5~50] & 44.2 \\
0.7 & 0.4783 & 0.4805 & 0.0022 & [0.5~20] & 19.78  \\
1 & 0.5645 & 0.5766 & 0.0121 & [0.5~6] & 6\\
1.3 & 0.6364 & 0.6435 & 0.0071 & [0.5~6] & 6 \\
1.6 & 0.6874 & 0.6924 &0.0050 & [0.5~6] & 6 \\
2 & 0.7347 & 0.7403 & 0.0056 & [0.5~6] & 6
\tame
\end{align*}
\caption{Upper and lower bounds on the $\mcal{L}_2$ norm of the LPV system defined in Section \ref{sec:catws}}
\label{tab:catwsresults}
\end{table}

\begin{figure}[h!]
\begin{center}
\includegraphics[width=12cm]{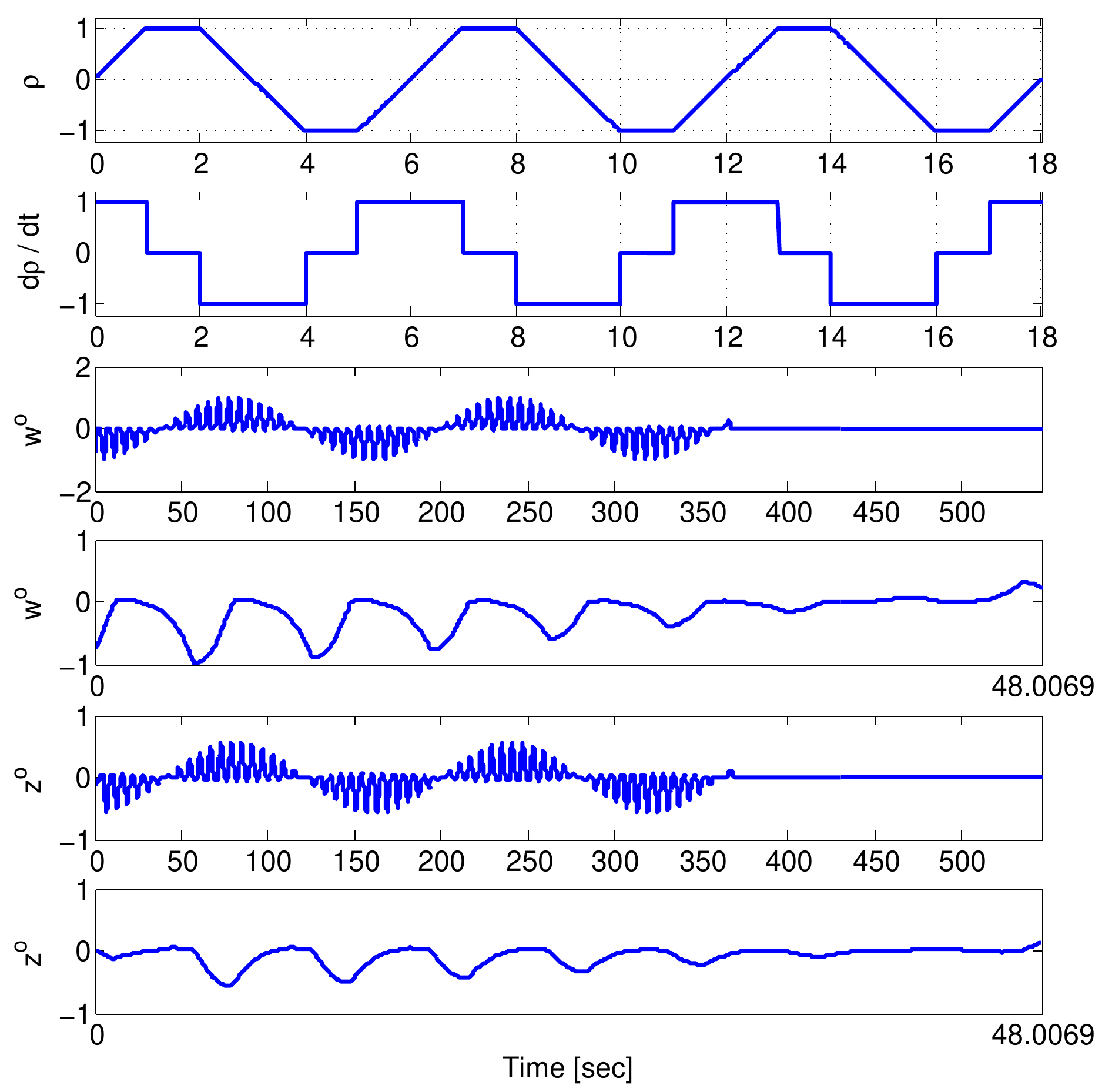}
\end{center}
\vspace*{-0.4cm}
\caption{Worst case scheduling trajectories and worst-case inputs/outputs at rate bounds $\overline\mu=1$ for the example in Section \ref{sec:catws}. The first two subfigures show the trajectory of the worst-case scheduling parameter $\rho(t)$ and its time derivative $\dot\rho(t)$.  In the 3rd and 4th subfigures $w^o$ is depicted over a longer and a shorter ($[0~8h]$) time interval. In the 5th and 6th figures $z^o$ is displayed over the same time intervals. }
\label{fig:catwswc}
\end{figure}

\begin{figure}[h!]
\begin{center}
\includegraphics[width=12cm]{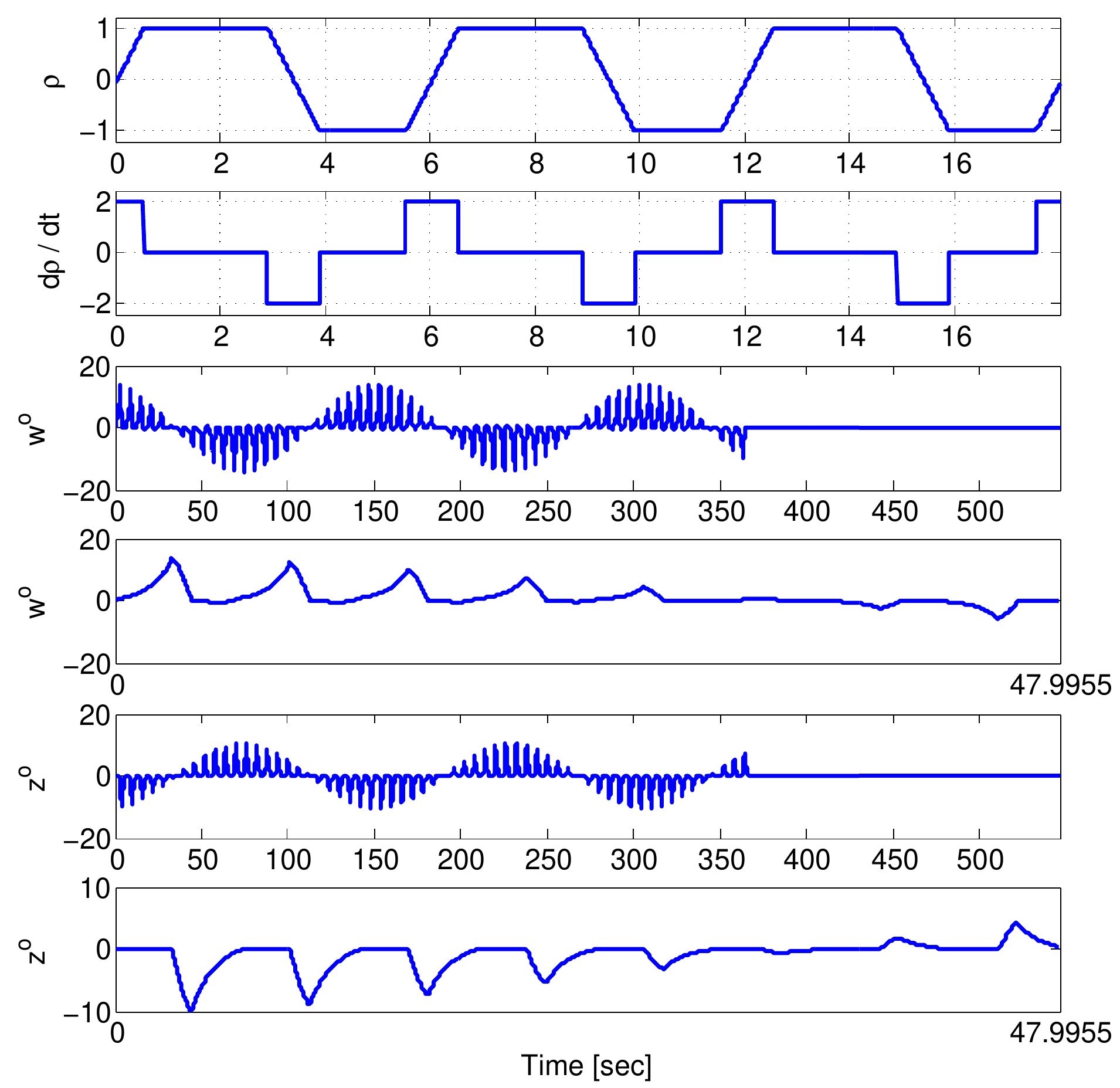}
\end{center}
\vspace*{-0.4cm}
\caption{Worst case scheduling trajectories and worst-case inputs/outputs at rate bounds $\overline\mu=2$  for the example in Section \ref{sec:catws}. The first two subfigures show the trajectory of the worst-case scheduling parameter $\rho(t)$ and its time derivative $\dot\rho(t)$.  In the 3rd and 4th subfigures $w^o$ is depicted over a longer and a shorter ($[0~8h]$) time interval. In the 5th and 6th figures $z^o$ is displayed over the same time intervals. }
\label{fig:catwswc}
\end{figure}

\begin{figure}[h!]
\begin{center}
\includegraphics[width=10cm]{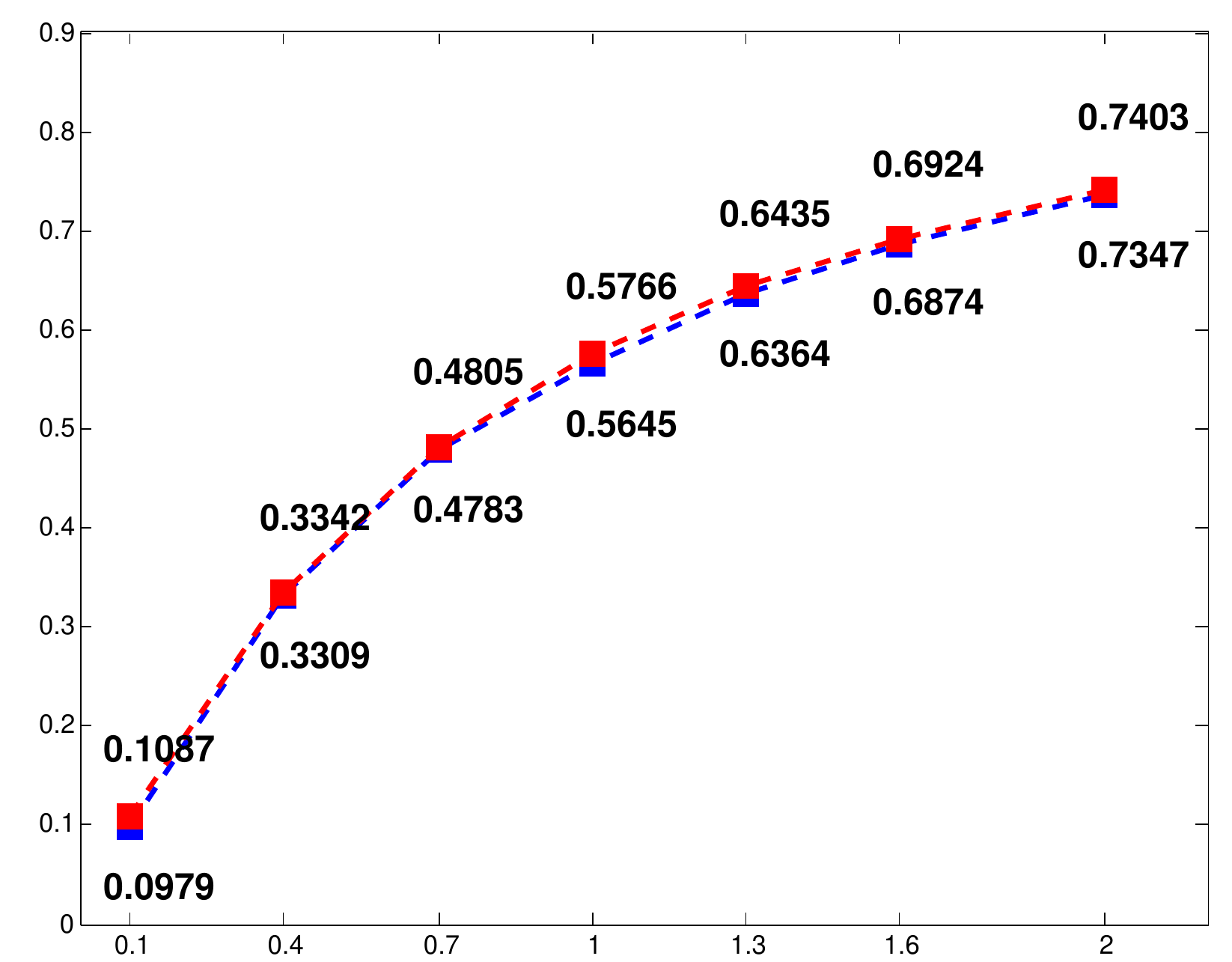}
\end{center}
\vspace*{-0.4cm}
\caption{Upper and lower bounds on the induced $\mcal{L}_2$ norm computed for the example in Section \ref{sec:catws}. }
\label{fig:catwsbounds}
\end{figure}

\subsection{Rotated LTI system} \label{sec:rotated}
The next example was constructed by taking a stable LTI system and performing a parameter dependent similarity transformation, which is actually a rotation, on the $A$ matrix. The system can be given in state-space form as follows:
\begin{align*}
\dot x &= R(\rho)^T\smat{cc}-0.5 & -0.4 \\ 3 & -0.5\tams R(\rho)x+\smat{c}0.5\\0.5\tams w \\
z&=[1~-1]x+0.1w
\end{align*}
where $R(\rho)=\ssmat{cc}\cos(\rho) & \sin(\rho) \\ -\sin(\rho) & \cos(\rho)\tamss$. The range and rate bounds for $\rho$ were chosen to be: $\rho\in [\pi/4~~\pi/2]$ and $\dot\rho\in[-0.1~~0.1]$. 
The upper bound on the induced $\mcal{L}_2$ norm was computed by using the following storage function and parameter grid:
\begin{align*}
&V(x,\rho)=x^T\left(P_0+\sum_{k=1}^{7} \rho_1^k P_k\right)x\\
&\Gamma=\{\rho_{1}=\pi/4,\ldots,\rho_{50}=\pi/2\}, \\ 
&\rho_{k+1}-\rho_{k}=\pi/196, ~~~k\in\{1\ldots 50\}
\end{align*}
We obtained  $\gamma_{ub}=3.3$. The lower bound computed at frozen parameter values over the grid above was $\gamma_{lb,fr}=2.696$. To compute the lower bound we chose the following derivative pattern for the scheduling trajectory: $r=0.1\cdot [1~ 0~ -1~ 0~ 1~ 0]$. The algorithm was allowed to tune the period length in the interval $[1~5]$. The lower bound we got is $\gamma_{lb}=3.15$, ($h=3.11$), which is again significantly larger than $\gamma_{lb,fr}$. The worst-case scheduling trajectories and the worst-case input are plotted in Fig. \ref{fig:rotatedwc}. 
(The worst-case input was constructed by starting from the unit modulus eigenvalue $-0.99997+ 0.00816i$. The parameter $K$ was set to $60$.) 
\begin{figure}
\begin{center}
\includegraphics[width=12cm]{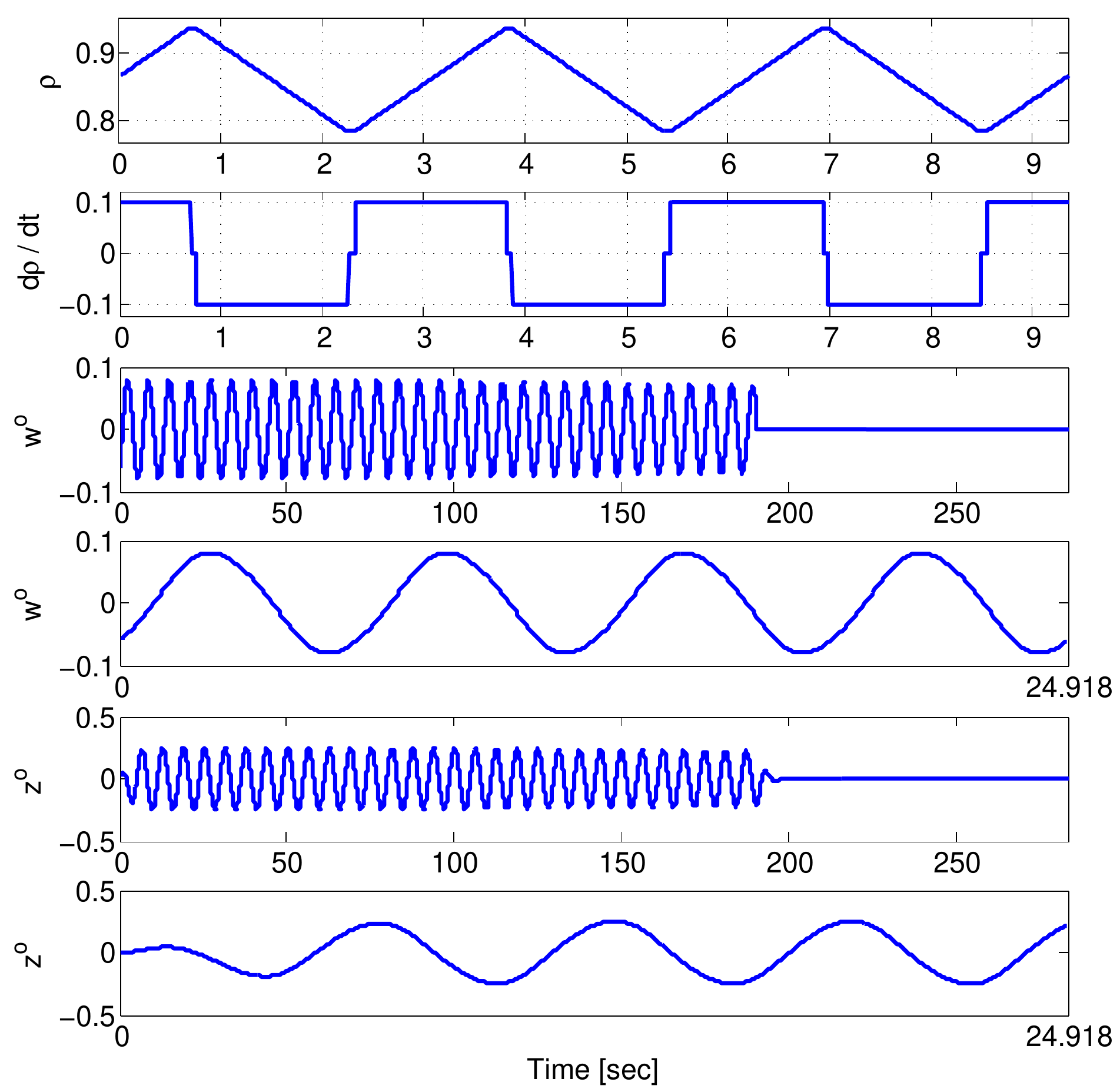}
\end{center}
\vspace*{-0.4cm}
\caption{Worst case scheduling trajectories and worst-case input/output for the LPV system in Section \ref{sec:rotated}. The first two figures show the trajectory of the worst-case scheduling parameter $\rho(t)$ and its time derivative $\dot\rho(t)$.  In the 3rd and 4th figures $w^o$ is depicted over a longer and a shorter ($[0~8h]$) time intervals. In the 5th and 6th figures $z^o$ is displayed over the same time intervals. }
\label{fig:rotatedwc}
\end{figure}

\subsection{A 2-parameter system}\label{sec:2parex}
The next, more complex example was constructed from the examples in sections \ref{sec:harald} and \ref{sec:catws}. We took two copies of the closed loop system \eqref{eq:haraldscloop} with scheduling parameter $\rho(t):=\rho_1(t)$, then we scaled the input of the first and the output of the second system by the same time-varying parameter $\rho_2(t)$ and defined the output as the difference between the outputs the two subsystems (like in Section \ref{sec:catws}). The state-space equations of the system obtained can be written as
\begin{align*}
\smat{c}\dot{x}_1\\ \dot{x}_2\tams &=\smat{cc}A(\rho_1)&0\\0&A(\rho_1)\tams \smat{c}x_1\\ x_2\tams+\smat{c}B(\rho_1)\\\rho_2B(\rho_1)\tams w \\
z&=\smat{cc}\rho_2C(\rho_1) & -C(\rho_1)\tams \smat{c}x_1\\ x_2\tams
\end{align*} 
The range and rate bounds were as follows: $2\leq\rho_1\leq7$, $-1\leq\rho_2\leq 1$ and $-1\leq\dot\rho_1,\dot\rho_2\leq 1$. The construction of the system implies that  $\gamma_{lb,fr}=0$. The upper bound on the induced $\mcal{L}_2$ norm was computed by using the following storage function and parameter grid:
\begin{align*}
&V(x,\rho)=x^T\left(P_0+\sum_{k=1}^{3} \rho_1^k P_k+\sum_{\ell=1}^5 \rho_2^{\ell}P_{\ell}+\frac{1}{\rho_1}P_9+\frac{1}{\rho_2}P_{10}\right)x\\
&\Gamma=\{\rho_{1,1}=2,\ldots,\rho_{1,30}=7\}\times \\
&~~~~~~~~~~~~~~~~~~~~~~~~~~~~~ \{\rho_{2,1}=-1,\ldots,\rho_{2,10}=1\},\\ 
&\rho_{k+1,1}-\rho_{k,1}=2/29,~~\rho_{\ell+1,2}-\rho_{\ell,2}=2/9,\\ 
&k\in\{1\ldots 30\},~~\ell=\{1\ldots 10\}
\end{align*}
We obtained  $\gamma_{ub}=5.38$. To compute the lower bound we chose piecewise linear scheduling trajectories with the following derivative patterns: $r_1=[1~ 0~ -1~ 0~ 1~ 0]$ and  $r_2=[0~ 1~ 0~ -1~ 0~ 1~ 0]$. The lower bound we got is $\gamma_{lb}=5.047$, ($h=16.2142$), which is again very close to the upper bound. The worst-case scheduling trajectories and the worst-case input are plotted in Fig. \ref{fig:haraldcatwswc}. 
(The worst-case input was constructed by starting from the unit modulus eigenvalue $-0.9997 + 0.0227i$. The parameter $K$ was set to $60$.)

\begin{figure}
\begin{center}
\includegraphics[width=12cm]{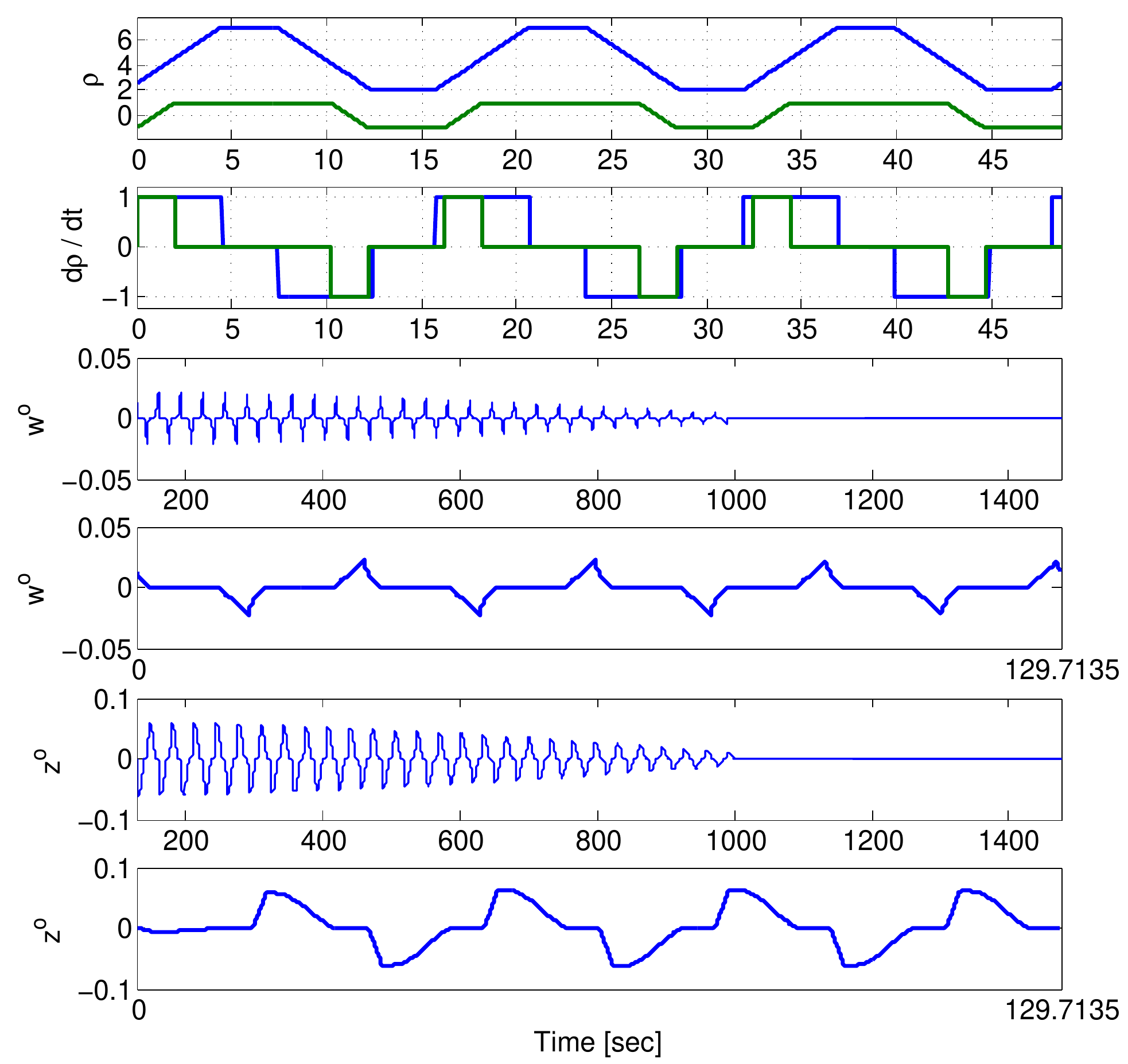}
\end{center}
\vspace*{-0.4cm}
\caption{Worst case scheduling trajectories and worst-case input/output for the LPV system in Section \ref{sec:2parex}. The first two figures show the trajectory of the worst-case scheduling parameter $\rho(t)$ and its time derivative $\dot\rho(t)$.  In the 3rd and 4th figures $w^o$ is depicted over a longer and a shorter ($[0~8h]$) time intervals. In the 5th and 6th figures $z^o$ is displayed over the same time intervals. }
\label{fig:haraldcatwswc}
\end{figure}

\section{Conclusion}\label{sec:conclusion}
In the paper a  numerical method is proposed for computing the lower bound of the induced $\mcal{L}_2$-gain of continuous-time, LPV systems. The algorithm finds this bound -- together with the worst-case parameter trajectory --  by using nonlinear optimization over periodic scheduling parameter trajectories. Restricting the domain of parameter trajectories to periodic signals enables to use the recent results for exact calculation of the $\mcal{L}_2$ norm for a periodic time varying system. It was shown that the proposed algorithm can be reliably implemented by standard numerical tools and provides precise approximation for the $\mcal{L}_2$ bound.


%

\appendices
\section{Proof of Lemma \ref{lem:innerprod}}
To prove the theorem we need the following properties of the state transition matrix:
for all $t, \tau, \tau' \in \mbb{R}$
\begin{enumerate}[(a)]
\item $\Phi(t,\tau)\Phi(\tau,t)=I$, which  implies $\Phi(t,\tau)=\Phi(\tau,t)^{-1}$ 
\item $\frac{d}{dt}\Phi(\tau,t)=-\Phi(\tau,t)A(t)$
\item $\frac{d}{dt}[\Phi(\tau,t)^*]=-A(t)^*\Phi(\tau,t)^*$.
\end{enumerate}
Item (a) follows from the definition and from the regularity \cite{bittanticolaneri09_psbook} of $\Phi(t,\tau)$; item (b) is a consequence of item (a) and can be proved by applying the chain rule of derivation to $0=\frac{d}{dt}\Phi(t,\tau)\Phi(\tau,t)$ and finally, item (c) is the direct consequence of item (b).

Then Lemma \ref{lem:innerprod} can be proved directly from the definition of the
adjoint (Equation \eqref{eq:adjLGss}) as follows:
\begin{align}
  \label{eq:TGinner}
  \langle (\hat{x}_h,\hat{z}), \mathbf{T}_G(x_0,w) \rangle 
  = \hat{x}_h^* x(h) + \int_0^h \hat{z}(t)^* z(t) \, dt
\end{align}
where $(x(h),z(t))$ are the final state and output of the PLTV system
in \eqref{eq:LGss}.  The solution of \eqref{eq:LGss} can be expressed in terms of the state
transition matrix as:
\begin{align}
  x(h) & = \Phi(h,0) x_0 + \int_0^h \Phi(h,\tau)B(\tau)w(\tau) \, d\tau \\
  z(t) & = C(t)\Phi(t,0) x_0 + \int_0^t C(t)\Phi(t,\tau)B(\tau)w(\tau) \, d\tau 
\end{align}
Use these relations to substitute for $x(h)$ and $z(t)$ in
Equation \eqref{eq:TGinner}. Re-arrange terms to obtain the following
form:
\begin{align}
\nonumber
&  \langle (\hat{x}_h,\hat{z}), \mathbf{T}_G(x_0,w) \rangle =
  \hat{x}(0)^* x(h) + \int_0^h \left[ B^*(\tau) \hat{x}(\tau) +
    D^*(\tau)\hat{z}(\tau)\right]^* w(\tau) \, d\tau
\end{align}
where we have defined the signal
\begin{align}
  \hat{x}(\tau) := \Phi(h,\tau)^* \hat{x}_h + \int_\tau^h
  \Phi(t,\tau)^* C(t)^* \hat{z}(t) \, dt
\end{align}
Using property (c) above, we obtain 
\begin{align*}
\frac{d\hat x(\tau)}{d\tau}=-A^*(\tau)\hat x(\tau)-C^*(\tau)\hat z(\tau),~~\hat x(h)=\hat x_h
\end{align*}
If we also define
$\hat{w}(\tau)=B^*(\tau)\hat x(\tau)+D^*(\tau)\hat z(\tau),$ and $\hat x_0=\hat x(0)$
then we can express the inner product $\langle (\hat x_h,\hat z), \mbf{T}_G(x_0,w) \rangle$ as
\begin{align*}
\hat x_0^*x_0+\int_0^h\hat{w}(\tau)^*w(\tau)d\tau=\langle \mbf{T}_G^*(\hat x_h,\hat z),(x_0,w)\rangle
\end{align*}
Thus the proof is complete.


\section*{Acknowledgment}
This work
was  supported by the National Science Foundation under
Grant No. NSF-CMMI-1254129 entitled ``CAREER: Probabilistic Tools for
High Reliability Monitoring and Control of Wind Farms''. Any opinions,
findings, and conclusions or recommendations expressed in this
material are those of the author and do not necessarily reflect the
views of the NSF.

The authors greatly acknowledge the help of Henrik Sandberg for making
the MATLAB code of the numerical example presented in
\cite{cantoni09_automatica} available for the purpose of this
research.

\ifCLASSOPTIONcaptionsoff
  \newpage
\fi



\bibliographystyle{IEEEtran}
\bibliography{../../../references/references}
\end{document}